\documentclass{article}
\usepackage{ijcai19}

\usepackage{times}
\usepackage{soul}
\usepackage{url}
\usepackage[hidelinks]{hyperref}
\usepackage[utf8]{inputenc}
\usepackage[small]{caption}
\usepackage{graphicx}
\usepackage{amsmath}
\usepackage{booktabs}
\usepackage{algorithm}
\usepackage{algorithmic}
\usepackage{graphicx}
\usepackage{subcaption}
\usepackage{pgf}
\usepackage{balance}  
\usepackage{amsfonts}
\usepackage{url}
\usepackage{algorithm}
\usepackage{algorithmic}
\usepackage{float}
\usepackage{amssymb,amsmath,epsfig,graphics,enumerate,float}
\usepackage{graphicx}
\usepackage{color}
\usepackage{xspace}
\usepackage{multirow}
\usepackage{amsthm}

\newcommand{\ignore}[1]{}
\newcommand{\nop}[1]{}
\newcommand{\eat}[1]{}
\newtheorem{definition}{Definition}
\newtheorem{lemma}{Lemma}

\newtheorem{example}{Example}

\newcommand{\squishlisttight}{
 \begin{list}{$\bullet$}
  { \setlength{\itemsep}{0pt}
    \setlength{\parsep}{0pt}
    \setlength{\topsep}{0pt}
    \setlength{\partopsep}{0pt}
    \setlength{\leftmargin}{2em}
    \setlength{\labelwidth}{1.5em}
    \setlength{\labelsep}{0.5em}
} }

\newcommand{\squishnumlist} {
\newcounter{qcounter}
\begin{list}{\arabic{qcounter}.~}{\usecounter{qcounter}} 
{  \setlength{\itemsep}{0pt}
    \setlength{\parsep}{0pt}
    \setlength{\topsep}{0pt}
    \setlength{\partopsep}{0pt}
    \setlength{\leftmargin}{2em}
    \setlength{\labelwidth}{1.5em}
    \setlength{\labelsep}{0.5em}
}}

\newcommand{\squishend}{
  \end{list}
}

\mathchardef\mhyphen="2D
\newcommand{\kw}[1]{{\ensuremath {\mathsf{#1}}}\xspace}
\newcommand{\kwnospace}[1]{{\ensuremath {\mathsf{#1}}}}
\newcommand{\stitle}[1]{\vspace{1ex} \noindent{\bf #1}}

\newcommand{\dkng}{$(k,d)$-neighborhood\xspace}

\DeclareMathOperator{\truss}{\tau}
\DeclareMathOperator{\newTruss}{\hat{\tau}}
\DeclareMathOperator{\low}{k_{low}}
\DeclareMathOperator{\up}{k_{up}}

\newcommand{\xin}[1]{{\color{black}{#1}}}
\newcommand{\soroush}[1]{{\color{black}{#1}}}

\newcommand{\baseline}{\kwnospace{Truss}-\kw{Decompostion}}

\newcommand{\vpp}{\kwnospace{Vertex}-\kw{PP}}
\newcommand{\epp}{\kwnospace{Edge}-\kw{PP}}
\newcommand{\hybridpp}{\kwnospace{Hybrid}-\kw{PP}}

\title{Fast Algorithm for K-Truss Discovery on Public-Private Graphs}

\author{
Soroush Ebadian$^{*,\dag}$
\And
Xin Huang$^{\dag}$
\affiliations
$^*$Sharif University of Technology\\ 
$^{\dag}$Hong Kong Baptist University
\emails
soroushebadian@gmail.com,
xinhuang@comp.hkbu.edu.hk
}

\begin{document}

\maketitle
\begin{abstract}
In public-private graphs, users share one public graph and have their own private graphs. A private graph \soroush{consists} of personal private contacts that only can be visible to its owner, e.g., hidden friend lists on Facebook and secret following on Sina Weibo. However, 
existing public-private analytic algorithms \soroush{have} not yet investigated the dense subgraph discovery of $k$-truss, where each edge is contained in at least $k - 2$ triangles. 
This paper aims at finding $k$-truss efficiently in public-private graphs. 
The core of our solution is a novel algorithm to update $k$-truss with node insertions. We develop a classification-based hybrid strategy of \xin{node insertions and edge insertions} 
to incrementally compute $k$-truss in public-private graphs. Extensive experiments validate the superiority of our proposed algorithms against state-of-the-art methods on real-world datasets.
\end{abstract}

\section{Introduction}\label{sec.intro}
Online social networks (e.g., Facebook, Twitter, Instagram, and Sina Weibo) have become vital platforms for connecting users to share information, post daily life events, and spread influence \cite{kempe2003maximizing,DBLP:conf/aaai/WilderIRT18,zhang2017finding,ijcai2018-507}. Due to \eat{the}\soroush{privacy} concerns, users tend to hide their connections, leading such private relationships not visible to other users in public but only themselves. For instance, Facebook users are likely to conceal their friend-list \cite{dey2012facebook}; Weibo users may prefer using the secret following feature, which hides their interested followees. Public-private graphs are developed to model this kind of social networks \cite{chierichetti2015efficient}. A public-private network contains a public graph which is visible and accessible to everyone; in addition, each vertex has a personal private graph only visible to its owner. Therefore, in the view of each user, the social network is a union of the public graph and its own private graph which can be significantly different for distinct users. Recently, many graph analytic tasks have been investigated on public-private networks, such as all-pairs shortest path distances, node similarities, and correlation clustering.

Dense subgraph discovery is a fundamental problem of many network analysis tasks. Numerous definitions of dense subgraphs have been proposed and investigated, e.g., clique, quasi-clique \cite{PeiJZ05}, $n$-clan \cite{mokken1979}, $n$-club \cite{mokken1979}, and $k$-plex \cite{xiao2017fast}. Recently, a popular notion of dense subgraphs that has been studied is $k$-truss. A $k$-truss is the largest subgraph of a graph such that each edge is contained in at least $k-2$ triangles within this subgraph. Finding $k$-trusses has many useful applications such as community search \cite{jiang2018vizcs}, complex network visualization \cite{zhao2012large}, 
 and task-driven team formation \cite{huang2016truss}. To the best of our knowledge, finding $k$-truss over public-private networks has not yet been studied in the literature. In this paper, we formulate the problem of finding public-private $k$-truss as follows. Given a query vertex \soroush{and parameter $k$}, the problem is to find $k$-truss in the public-private graph owned by this query vertex.

\begin{figure}[t]
    \centering 
    \vspace{-0.4cm}
    \includegraphics[scale=0.65]{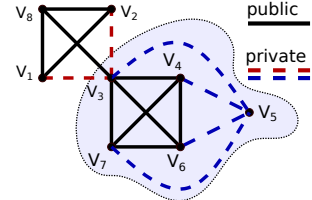}
    \vspace{-0.3cm}
    \caption{An example of public-private graph. Black solid edges are public. \xin{Blue dashed edges are private to $v_5$, and red dashed edges are private to $v_3$. The gray area is a $5$-truss in personalized pp-graph of $v_5$ as $g_{v_5}$.}
    }
    
    \label{fig:pp-formulation-example}
    \vspace{-0.2cm}
\end{figure}

Efficient extraction of public-private $k$-truss raises significant challenges. A straightforward approach is to ignore users' private edges, which can lead to inaccurate results. Another approach is to apply 
\soroush{truss decomposition on the public-private graph of the query vertex to extract the $k$-truss with the given parameter $k$.} However, this method computes $k$-truss from scratch, which is particularly inefficient for large-scale networks. To tackle these challenges, we develop an index-based computational paradigm to efficiently update
\soroush{a truss index}
on a public graph using the edges in a private graph, with a minimal amount of recomputation on the public graph. 

To summarize, we make the following contributions:

\squishlisttight
\item We formulate a new problem of finding $k$-truss over public-private graphs, that is finding $k$-truss in public-private graph owned by a given query vertex (Section~\ref{sec.problemdef}).

\item We analyze the structural properties of $k$-truss on public-private networks. Based on the observations, we develop $k$-truss updating algorithm using a hybrid strategy of node/edge insertions/deletions (Section~\ref{sec.fast}).

\item 
We validate the efficiency of our proposed methods through extensive experiments on real-world datasets of public-private networks (Section~\ref{sec.exp}). 

\end{list} 


\section{Related Work}\label{sec.related}

\stitle{Public-private graph processing.} Several essential problems of graph analysis on public-private graphs have been studied in \cite{chierichetti2015efficient,archer2017indexing}, such as the size of reachability tree \cite{cohen2007summarizing}, all-pairs shortest paths \cite{das2010sketch}, pairwise node similarities \cite{haveliwala2002topic}, correlation clustering \cite{bansal2004correlation}.
Moreover, the public-private model of data summarization has been investigated and solved by a fast distributed algorithm \cite{mirzasoleiman2016fast}. 

\stitle{K-truss mining.} Recently, several studies on $k$-truss mining \soroush{have} been investigated \cite{cohen2008,zhang2018finding}. Equivalent concepts of $k$-truss termed as different names include
triangle $k$-core \cite{YZhang12}, $k$-dense community \cite{saito2008extracting,gregori2011k}, and $k$-mutual-friend subgraph \cite{zhao2012large}. Truss decomposition is to find the non-empty $k$-truss for all possible $k$ values in a graph. Algorithms of truss decomposition have also been studied in 
different types of graphs (e.g., directed graphs \cite{takaguchi2016cycle}, uncertain graphs \cite{zou2017truss}, and dynamic graphs \cite{YZhang12,huang2014querying}).  

In contrast to the above studies, finding $k$-truss over public-private networks is studied for the first time in this paper.

\section{Preliminary}
\label{sec.problemdef}

We consider a simple and undirected graph $G=(V, E)$ where $V$ and $E$ are the vertex set and  edge set respectively.
We define $N(v)=\{u\in V: (v,u)\in E\}$ as the set of neighbors of a vertex $v$, and $d(v)=|N(v)|$ as the degree of $v$ in $G$. 
For a set of vertices $S\subseteq V$, the induced subgraph of $G$ by $S$ is denoted by $G[S]$, where the vertex set is $S$ and the edge set is $E(G[S]) = \{(v,u)\in E: v, u\in S\}$. 

\eat{Table \ref{tab:notations} holds frequently used notations in the paper.

\begin{table}[t]
	\scriptsize
	\centering
	\begin{tabular}{ll}  
		\toprule
		Notation & Description\\
		
		\midrule
		$G = (V(G), E(G))$ & An undirected simple graph $G$\\ 
		$N(v)$ & Set of neighbor vertices to vertex $v$\\
		$E(v)$ & Set of edges incident to vertex $v$\\
		$E_v$ & Set of private edges incident to vertex $v$\\
		$\sup_H(v)$ & \textit{Support} of edge $e$ in $H$\\
		$\truss(e), \truss(v)$ & Trussness of edge $e$, vertex $v$\\
		$\newTruss(e), \newTruss(v)$ & Trussness of edge $e$ and vertex $v$ after an update\\
		$G_S$ & Induced subgraph of graph $G$ over vertex set $S$\\
		$G^{k, d}_v$ & $(k, d)$-Neighborhood of vertex $v$ in $G$\\
		\bottomrule
	\end{tabular}
	\caption{Frequently used notations}
	\label{tab:notations}
\end{table}
}

\subsection{Public-Private Graphs}  
We first introduce a model of public-private graph  $\mathcal{G}$ \cite{chierichetti2015efficient}. A public-private graph $\mathcal{G}$ \soroush{consists} of one public graph and multiple private graphs. Given a public graph $G=(V,E)$, the vertex set $V$ represents users, and the edge set $E$ represents connections between users. For each vertex $u$ in the public graph $G$, $u$ has an associated private graph $G_u=(V_u, E_u)$, where $V_u \subseteq V$ are the users from public graph and the edge set $E_u$ satisfies $E_u \cap E=\emptyset$. The public graph $G$ is visible to everyone, and the private graph $G_u$ is only visible to user $u$. Thus, in the view of user $u$, she/he can see and access the structure of graph that is the union of public graph $G$ and its own private graph $G_u$, i.e., $G\cup G_u = (V, E\cup E_u)$ \cite{huang2018pp}. 
The personalized public-private graph (a.k.a. pp-graph in short) owned by a vertex $u$ is defined as follows.

\begin{definition}
[Personalized PP-Graph]
\label{def.ego}
Given a public-private graph $\mathcal{G}$ and a vertex $u$, the personalized pp-graph of $u$ is denoted by $g_u$, where $g_u = G\cup G_u = (V, E\cup E_u)$. Here, $E_u$ are the private edges only visible to $u$, and $E\cap E_u = \emptyset$.
\end{definition}

\subsection{K-Truss}
A triangle  is a cycle of length 3 in graphs. Given three vertices $u, v, w \in V$, the triangle formed by $u, v, w$ is denoted by $\triangle_{uvw}$. The support of an edge is defined as follows. 

\begin{definition}
[Support] Given a subgraph $H\subseteq G$, the support of an edge $e=(u,v)$, denoted by $\sup_H(e)$, is defined as the number of triangles containing edge $e$ in $H$, i.e.,  $\sup_H(e)=|\{\triangle_{uvw} : (u,v), (u, w), (v, w) \in E(H)\}|$. 
\label{def.support}
\end{definition}

We drop the subscript and denote the support as $\sup(e)$, when the context is obvious. Based on the support, we give a definition of $k$-truss \cite{WangC12} as follows.

\begin{definition}
[K-Truss] A $k$-truss $H$ of graph $G$ is defined as the largest subgraph of $G$ such that every edge $e$ has support of at least $k-2$ in this subgraph, i.e., $\sup_H(e)\geq k-2$. 
\label{def.ktruss}
\end{definition}

\subsection{Problem Statement}
The problem of public-private $k$-truss discovery studied in this paper is formulated as follows.

\stitle{Problem formulation}: Given a public-private graph $\mathcal{G}$, a vertex $u \in V$, and an integer $ k \geq 2$, the problem is to find the $k$-truss in the personalized pp-graph $g_{u}$ where $g_{u} = G\cup G_u$. 

\begin{example}
Consider \soroush{the}\eat{a} public-private graph $\mathcal{G}$ in Figure \ref{fig:pp-formulation-example}, a query vertex $v_5$, and $k = 5$. Black  edges are public. Blue edges are private to $v_5$. The answer of $5$-truss in personalized pp-graph $g_{v_5}$  is the subgraph depicted in the gray region. 
\end{example}

\section{Proposed Algorithms}\label{sec.fast}

This section introduces our algorithms for finding $k$-truss in personalized pp-graph $g_u$, w.r.t. a query vertex $u$. We first give an overview of our ideas in Section \ref{sec.idea}, and then present a well thorough description of technical details afterward. 

\subsection{Overview of Algorithmic Framework}
\label{sec.idea}
We consider two different ideas. 

\stitle{Solution 1: online search algorithm.} One intuitive approach is to apply truss decomposition on pp-graph $g_u$ to iteratively remove edges with less than $k-2$ triangles and output the remaining graph as answers. However, such computing $k$-truss from scratch on $g_u$ for each query vertex $u$ is obviously inefficient for big graphs with a large number of vertices.

\stitle{Solution 2: index-based search algorithm.} Recall that personalized pp-graph  $g_u=G\cup G_u$ has a public graph $G$  and a private graph $G_u$ only available to $u$. The public graph $G$ is available to everyone, and the structure of $G$ is identical to each query vertex $u$. The idea of index-based search algorithms is to construct a structural index of public graph $G$ offline, and then online find $k$-truss based on the precomputed index of $G$ and additional graph $G_u$. In the following, we introduce a concept of trussness, which is useful for constructing the truss-index for $k$-truss discovery.

\begin{definition}
[Trussness] Given a subgraph $H\subseteq G$, the trussness of $H$ denoted by $\tau(H)$ is defined as the minimum support of edges in $H$ plus 2, i.e.,  $\tau(H) =\min_{e\in E(H)}$ $\{\sup_H(e)+2\}$. The trussness of an edge $e\in H$ denoted by $\tau_H(e)$ is defined as the largest number $k$ such that there exists a connected $k$-truss $H'$ containing $e$, i.e., $$\tau_H(e) = \max_{H'\subseteq H, e\in E(H')} \tau(H').$$
\label{def.edgetruss}
\vspace{-0.3cm}
\end{definition}

We drop the subscript and denote $\tau_G(e)$ as $\tau(e)$ when the context is obvious. According to Def.~\ref{def.edgetruss}, $k$-truss of $G$ is the union of all edges $e$ with $\tau(e)=k$. The truss-index of public graph $G$ keeps the trussness of all edges in $G$. Given a truss-index of $G$, the remaining issue is \emph{how to update the truss-index for pp-graph $G\cup G_u$, w.r.t. the additional $G_u$.}

\stitle{Updating truss-index using edge insertions.} A simple approach is to add edges of $G_u$ one-by-one into $G$ and update the truss-index accordingly, by using an existing edge-insertion algorithm \cite{huang2014querying}. However, when the number of private edges $E_u$ is large, the adaptation of edge-insertion may be inefficient. For example, consider the example graph $\mathcal{G}$ in Figure~\ref{fig:pp-formulation-example} and query vertex $v_5$ with 4 private edges. It invokes the edge-insertion algorithm for 4 times.

\stitle{Our approach.} 
To address the above issue, we propose a batch-update algorithm using node-insertion. The idea is to simultaneously insert a new node $u$ with all its incident edges into graph $G$ at the same time, and call node-insertion algorithm only \soroush{once}\eat{call}. In \soroush{the} above example, it needs only one node-insertion of \soroush{the} isolated node $v_5$ and all of its private edges.  To handle the truss-index update with node insertions efficiently, the key is to identify the affected region in the graph precisely. We provide a theoretical analysis to define the affect\soroush{ed} scope in Section \ref{sec.theoretical}, and the detailed algorithm of node-insertion in Section~\ref{sec.insertionalg}.
\soroush{However, when $u$ has both public and private edges, we can first remove $u$ with its public edges, and re-insert it with both public and private edges incident to $u$ using node-insertion. This method has significant advantages outperforming edge-insertion when private edges of $u$ are much larger than public edges of $u$. On other cases, edge-insertion method may perform better.}
Therefore, we construct a classifier to determine which algorithms of node-insertion and edge-insertion should be applied. This classification-based hybrid approach is developed to 
fast 
find $k$-truss in pp-graph $g_u$, which is presented in Section~\ref{sec.ppquerying}.

\subsection{Theoretical Analysis}
\label{sec.theoretical}
In this section, we present useful rules for truss-index updating with node insertion/deletion. Consider a vertex $v$ and the set of edges incident to $v$ as $E(v)$. 
In the case of node insertion, we insert a new vertex $v$ and its incident edges $E(v)$ into $G$, where $E(v)\cap E(G) = \emptyset$; in the case of node deletion, we delete vertex $v$ and all its incident edges $E(v)$ from $G$, where $E(v)\subseteq E(G)$. We use $\truss(e)$ and $\newTruss(e)$ to denote trussness of edge $e$ before and after updating operation. Motivated by 
\cite{huang2014querying}, the following three updating rules hold.

\stitle{Rule 1:}
If new node $v$ is inserted into graph $G$ with $\newTruss(v) = \max_{e \in E(v)} \newTruss(e) = l$, then  $\forall e \in E(G)$ with $\truss(e) \ge l$, $\newTruss(e) = \truss(e)$ holds.

\stitle{Rule 2:}
If node $v$ is deleted from graph $G$ with $\truss(v) =  \max_{e \in E(v)} \truss(e) = l$, then  $\forall e \in E(G) \setminus E(v)$ with $\truss(e) > l$, $\newTruss(e) = \truss(e)$ holds.

\stitle{Rule 3:}
$\forall e \in E(G)\setminus E(v)$, $|\newTruss(e) - \truss(e)| \le 1$ holds. 

Rules 1 and 2 hold because node $v$ is not present at any $(l + 1)$-truss subgraph. Rule 3 holds because for each edge, at most one triangle will be formed/deformed after one node insertion/deletion, hence $\sup(e)$ will change at most by one.


In the following, we focus on node insertions. In order to apply Rule 1 for pruning, the value of $\newTruss(v)$ is required. However, an exact computation of $\newTruss(v)$ is costly expensive. Instead, we develop another rule based on an upper bound of $\newTruss(v)$ below.

\stitle{Rule 1':}
If node $v$ is inserted into graph $G$ and $\overline{\newTruss(v)} \ge \newTruss(v)$, then  $\forall e \in E(G)$ with $\truss(e) \ge \overline{\newTruss(v)}$, $\newTruss(e) = \truss(e)$ holds.

To desire an upper bound of $\newTruss(e)$, we need a new definition of \dkng as follows.

\begin{definition}[\dkng]
Given a graph $G$, \dkng of vertex $v$, denoted by $G^{k, d}_v$, is the maximal subgraph $H \subseteq G[N(v)]$ holding
\label{def.dkneigh}
\begin{enumerate}
	\item $\truss_G(e) \ge k, \forall e \in E(H)$ and
	\item $d_H(u) \ge d, \forall u \in V(H)$.
\end{enumerate}
\end{definition}

\begin{lemma}
\label{lem.bound}
Consider a new node $v$ and its incident edges $E(v)=\{(v,w): w\in N(v)\}$ are inserted into graph $G$. For each new edge $e=(v, w)$ in the new graph $G_{new}$, the trussness of $e$, $\newTruss(e)$, satisfies $\low(e) \le \newTruss(e) \le \up(e)$ where 
$$\low(e) = \max \{k: w \in G^{k, k - 2}_v \}$$ 
and 
$$\up(e) = \max \{k:w \in G^{k-1,k-2}_v\}.$$ 
Moreover, $|\up(e) - \low(e)| \le 1$ holds.
\end{lemma}

\begin{proof}
We consider an edge $e^*=(v, w^*)$ in $G_{new}$. For simplification, we denote by $\low(e^*) =k_l$ and $\up(e^*) =k_u$.

First, we prove $\newTruss(e^*) \geq \low(e^*) =k_l$. To prove it, we show that there exists a $k_l$-truss $H^*$ of $G_{new}$ containing $e^*$. By the definition of $\low(e)$, there exists a $k_l$-truss  $H$ of $G$, i.e., $\forall e \in E(H)$, $\sup_H(e)\geq k_l-2$.  
Let $H^* = (V(H)\cup \{v\}, E(H) \cup \{(w, v) \mid w \in G^{k_l, k_l - 2}_v\})$, 
which adds vertex $v$ and \eat{the}$v$'s incident edges $(w,v)$ with $\low((w, v)) \geq k_l$ into $H$. 
For each edge $(w, v)\in E(H^*)\setminus E(H)$, $w$ and $v$ have at least $k_l-2$  common neighbors  in $H^*$ by the second condition of Def.~\ref{def.dkneigh}, indicating $\sup_{H^*}((w,v))\geq k_l-2$; moreover, for each edge 
$\soroush{e}\eat{(u, w)} \in E(H^*)\cap E(H)$
, $\sup_{H^*}(e)\geq \sup_H(e)\geq k_l-2$. As a result, $H^*$ is a $k_l$-truss, and  $\newTruss(e^*) \geq \tau_{H^*}(e^*) \geq k_l$. 

Second, we prove  $\newTruss(e^*) \leq \up(e^*) =k_u$ by contradiction. Assume that $\newTruss(e^*) \geq k_u+1$, there exists a $(k_u+1)$-truss $H$ containing $e^*$ in $G_{new}$. We delete the node $v$ and all its incident edges $(v, w)$ from $H$, which leads to a new graph $H^*$.  By Rule 3, the trussness of each edge $e$ in $H^*$ decreases by at most 1 after the node deletion of $v$, i.e., $\tau_{H^*}(e) \geq k_u$. 
Let the vertex set $S= V(H) \cap N(v)$.  Obviously, $H^*[S]$ = $H[S]$. 
For each edge $e$ in $H^*[S]$, $\tau_{G}(e) \geq \tau_{H^*}(e) \geq k_u$; for each vertex $w$ in $H[S]$, 
the edge $(v,w)$ belongs to $(k_u+1)$-truss $H$, 
indicating $w$ has at least $k_u-1$ neighbors in $H$ and also in $H[S]$, i.e.,  $d_{H^*[S]}(w) = d_{H[S]}(w)$ $\geq k_u-1$. By Def.~\ref{def.dkneigh}, $H^*[S]$ is a $(k_u, k_u-1)$-neighborhood as $G_v^{k_u, k_u-1}$ in $G$. Thus, $\max \{k:w \in G^{k-1,k-2}_v\} \geq k_u+1$.
However, by the definition of $\up(e)$, $k_u = \max \{k:w \in G^{k-1,k-2}_v\}$ and  $k_u\geq k_u+1$, which is a contradiction.

Third, we prove   $|k_u - k_l|\leq 1$. 
Obviously, $k_l \leq k_u$ and $k_l \leq k_u+1$. We next prove $k_u \leq k_l+1$. According to the definition of $\up(e)$, we have $k_u =$ $\max \{k:w \in G^{k-1,k-2}_v\}$ $\leq$ $\max \{k:w \in G^{k-1,k-3}_v\}$. Moreover, $\max \{k:w \in G^{k-1,k-3}_v\} =$ $\max \{k:w \in G^{k,k-2}_v\}+1$ $= k_l+1$ by the definition of $\low(e)$, and we derive $k_u\leq k_l+1$. As a result, $|k_u - k_l|\leq 1$.
\end{proof}

\stitle{Scope of Affected Edges.}
Let $\overline{\newTruss(v)} = \max\{ \up(e): e \in E(v)\}$ be an upper bound of $\newTruss(v)$, and the weight of a triangle be the minimum trussness of edges within this triangle.

\begin{enumerate}
\item \stitle{Node Insertion.} Edge $e = (x, y) \in E(G) \cup E(v)$ with $\tau(e) < \overline{\newTruss(v)}$, may have trussness increment if $(v, x, y)$ form a triangle of weight $\tau(e)$, or  $e$ is connected to  $v$ through a series number of adjacent triangles each with weight of $\tau(e)$.

\item \stitle{Node Deletion.} Edge $e = (x, y) \in E(G) - E(v)$ with $\tau(e) \le \max\{\tau(e): e \in E(v)\}$ may have trussness decrement if $(v, x, y)$ form a triangle of weight $\tau(e)$ or $e$ is connected to $v$ through a series number of adjacent triangles each with weight of $\tau(e)$.
\end{enumerate}

\begin{algorithm}[t]
\small
\vspace{-0.05cm}
\caption{Node-Insertion Updating Algorithm}
\label{algo:nodeinsertion}
\begin{flushleft}
\vspace*{-0.1cm}
\textbf{Input:} $G=(V, E)$, new node $v$, edge set $E(v)$\\
\textbf{Output:} $\newTruss(e)$ for each $e \in E \cup E(v)$\\
\end{flushleft}
\begin{algorithmic}[1]
\vspace*{-0.1cm}
\STATE $G \leftarrow G \cup (v, E(v))$

\STATE Compute $\low(e), \up(e)$ for all $e \in E(v)$ by Algorithm \ref{algo:nodebound}

\STATE \textbf{for} $e$ in $E(v)$ \textbf{do}

\STATE \hspace{0.3cm} $\truss(e) \leftarrow \low(e)$

\STATE \hspace{0.3cm} \textbf{if} $\low(e) < \up(e)$ \textbf{then}

\STATE \hspace{0.6cm} $L_{\low(e)} \leftarrow L_{\low(e)} \cup \{e\}$

\STATE \textbf{for} $e = (u, w)$ in $G_{N(v)}$ \textbf{do}

\soroush{
\STATE \hspace{0.3cm} \textbf{if} $\truss(e) < \min\{\up((u, v)), \up((w, v))\}$ \textbf{then}
	
\STATE \hspace{0.3cm} \hspace{0.3cm} $L_{\truss(e)} \leftarrow L_{\truss(e)} \cup \{e\}$
}

%
%

\STATE $k_{max} \leftarrow \max\{\up(e) : e \in E(v)\}$ \textbf{then}

\STATE \textbf{for} $k \leftarrow k_{max} - 1$ to $2$  \textbf{do}

\STATE \hspace{0.3cm} UpdateTrussness($k$, $L_k$)

\end{algorithmic}
\end{algorithm}

\subsection{Node-Insertion Updating Algorithm}
\label{sec.insertionalg}

In this section, we propose a algorithm to update \soroush{the truss-index with node insertions.}

\begin{algorithm}[t]
\small
\vspace*{-0.05cm}
\caption{Node-Insertion Bound Computing Algorithm}
\label{algo:nodebound}
\begin{flushleft} 
\vspace*{-0.1cm}
\textbf{Input:} $G=(V, E)$, new node $v$, edge set $E(v)$, $type \in \{$low, up$\}$\\
\textbf{Output:} $\{\operatorname{k_{type}}(e) \colon e \in E(v)\}$ trussness bound according to $type$
\end{flushleft}
\begin{algorithmic}[1]
\vspace*{-0.1cm}
\STATE $H \leftarrow G[N(v)]$; $k \leftarrow 2$

\STATE \textbf{while} $H \neq \emptyset$  \textbf{do}

\STATE \hspace{0.3cm} \textbf{while} $\exists e\in E(H)$ with $\tau_H(e) < k$ \textbf{do}

\STATE \hspace{0.6cm} Delete edge $e$ from $H$;

\STATE \hspace{0.3cm} \textbf{while} $\exists d_H(u) < (k - 2 $ \textbf{if} $type =$ low \textbf{else} $k - 1)$ \textbf{do}

\STATE \hspace{0.6cm} Delete vertex $u$ and its incident edges from $H$;

\STATE \hspace{0.6cm} $\operatorname{k_{type}}((v, u)) \leftarrow (k - 1 $ \textbf{if} $ type = $ low \textbf{else} $ k)$

\STATE \hspace{0.3cm} $k \leftarrow k + 1$
\end{algorithmic}

\end{algorithm}

\stitle{Node-insertion Algorithm}. Algorithm \ref{algo:nodeinsertion} updates the truss-index with inserting node $v$ and its incident edges $E(v)$ to $G$. Lower and upper bounds of each edge can be computed by calling Algorithm \ref{algo:nodebound} (line 2). According to the scope of affected edges, first, trussness of newly added edges is set to $\low(e)$, and then candidate edges for updating are found through lines 3-10.
Newly added edges with $\low(e) + 1 = \up(e)$ might have trussness increase which are found in lines 3-6. Furthermore, any edge in \soroush{$G[N(v)]$} that might get affected is found through lines 7-9.
According to Rule 1', $k_{max}$ is set to maximum of the upper bounds which results in pruning all unaffected edges which have trussness of at least $k_{max}$. The procedure of level-by-level updating truss-index (line 12) follows the edge-insertion algorithm \cite{huang2014querying}.

\stitle{Computing $\up(e)$ and $\low(e)$}. Algorithm \ref{algo:nodebound} computes the upper bound $\up(e)$ and lower bound $\low(e)$ in Lemma \ref{lem.bound}. Computing lower bounds and upper bounds are almost the same, and the same code can be used with passing a parameter $type \in \{\text{low}, \text{up}\}$ to indicate which bound to compute. Consider the case where $type$ is low and lower bound is required. Algorithm starts with an induced subgraph $G[N(v)]$ as $H$, which is $G^{2, 0}_v$. In each iteration refines $H$ to reach $G^{k, k - 2}_v$ from $G^{k - 1, k - 3}_v$. Assume the $(k-1)$-th step has correctly stored $G^{k - 1, k - 3}_v$ in $H$. Algorithm first removes edges with trussness less than $k$ (lines 3-4), and then removes vertices with degree less than $k - 2$ (lines 5-7). All removals have been necessary and obtained $H$ is $G^{k, k - 2}_v$. Each node removed in this iteration is member of $G^{k - 1, k - 3}_v$ but not $G^{k, k - 2}_v$; therefore the bound $k-1$ finally found by the algorithm is the maximum possible value.

\subsection{A Classification-based Hybrid Algorithm for Finding K-Truss in Public-Private Graphs}
\label{sec.ppquerying}

This section introduces a classification-based algorithm of updating truss-index from public graph $G$ to personalized pp-graph $G \cup G_u$. The algorithm is outline\soroush{d} in Algorithm \ref{algo:hybridpp}, which uses a hybrid strategy of updating with node/edge insertions/deletions. Specifically, we have two following strategies to update truss-index.

\squishlisttight
\vspace{-0.1cm}
\item \stitle{Edge-PP.} Add private edges of $E(G_u)$ one by one into $G$ using edge-insertion algorithm \cite{huang2014querying}.

\item \stitle{Vertex-PP.} Remove vertex $u$ with its public edges by node-deletion algorithm, then add back $u$ with all incident edges
\soroush{to obtain} $G\cup G_u$ using node-insertion algorithm in Algorithm \ref{algo:nodeinsertion}.
\end{list}

\begin{figure}[h]
	\centering 
	\vspace{-0.3cm}
	\includegraphics[width=0.9\columnwidth]{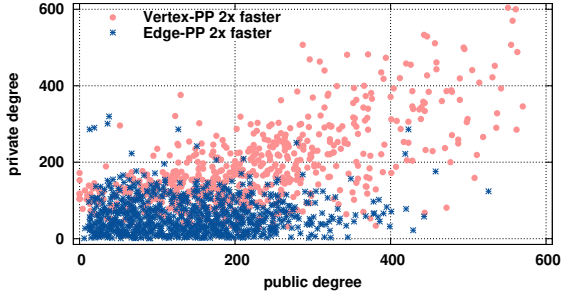}
	\vspace{-0.3cm}
	\caption{Win case distribution of \vpp and \epp for 1367 random query nodes in PP-DBLP-2013.}
	\label{fig:win-case-comparison}
	\vspace{-0.2cm}
\end{figure}
\begin{algorithm}[t]
\small
\vspace{-0.05cm}
\caption{Hybrid-PP Algorithm}
\label{algo:hybridpp}
\begin{flushleft}
\vspace*{-0.1cm}
\textbf{Pre-process:}\\
\textbf{Input:} $G=(V, E)$, truss-index $\{\tau(e) \colon \forall e \in E\}$\\
\textbf{Output:} Classification model $C: D \rightarrow \{C_V, C_E\}$ \\
\end{flushleft}

\begin{algorithmic}[1]
\vspace*{-0.1cm}
\STATE Training vertex set $S \leftarrow sample\_nodes(G)$
\STATE \textbf{for} $v$ in $S$ \textbf{do}
\STATE \hspace{0.2cm} $T_{V}(v), T_{E}(v) \leftarrow$ Runtime of \vpp and \epp on $v$
\STATE $\textbf{X}=[feature(v) : {v \in S}]$
\STATE $\textbf{Y}=[C_V $ \textbf{if} $T_{V}(v) < T_{E}(v)$ \textbf{else} $C_E : v \in S]$
\STATE $C \leftarrow$Classifier-Construction$(\textbf{X}, \textbf{Y})$
\end{algorithmic}

\begin{flushleft} 
\vspace*{-0.1cm}
\textbf{Query:}\\
\textbf{Input:} query node $u$, private graph $G_u$, and integer $k$\\
\textbf{Output:} the $k$-truss in $G \cup G_u$\\
\end{flushleft}

\begin{algorithmic}[1]
\vspace*{-0.1cm}
\STATE \textbf{if} $C.predict(feature(u)) = C_V$ \textbf{then}
\STATE \hspace{0.3cm} Update index using \vpp($u$)
\STATE \textbf{else} Update index using \epp($u$)
\STATE \textbf{return} Query\_KTruss($u$, $k$) on updated index
\end{algorithmic}
\end{algorithm}
\noindent

Note that \vpp can not directly add $u$'s private edges into $G$,
\soroush{as Lemma \ref{lem.bound} holds only}
for the insertion of a completely new vertex. Both algorithms update the truss-index correctly; however, the optimal choice between \epp and \vpp is \soroush{not} straightforwardly clear, because multiple aspects affect efficiency. Determining efficiency performance  requires a global knowledge of the whole graph structure. Local topological properties is not sufficient to decide which algorithm works faster. For example, Figure \ref{fig:win-case-comparison} shows the distribution of cases in which \vpp or \epp perform at least two times faster than the other algorithm; it consists of 1367 randomly selected nodes w.r.t their public and private degrees sampled from PP-DBLP-2013 dataset \cite{huang2018pp}. It is clear by Figure \ref{fig:win-case-comparison} that the simple distinction of degree is not sufficient to make good decision upon which algorithm to use. Thus, we formulate and tackle the problem of using  \epp and \vpp as a classification task.


Algorithm \ref{algo:hybridpp} predicts which updating method between \vpp and \epp works faster in terms of the given query, and runs that algorithm to answer the query.
There are two classes $C_V$ and $C_E$ that each node $u \in V$ is in class $C_V$ if \vpp works faster than \epp to update index from public graph $G$ to $G \cup G_u$ and the similar for $C_E$. We present each node $u$ using the following features: 1) public degree of $u$; 2) private degree of $u$; 3) the number of triangles containing $u$ respectively in public graph $G$, private graph $G_u$, and pp-graph $g_u$; 4) the trussness sum of public edges; 5) the maximum trussness of public edges; and 6) the summation and maximum over lower and upper bounds calculated by Algorithm \ref{algo:nodebound}. Let $D$ denote the vector space consisted of features described above, and $feature: V \rightarrow D$ the mapping function from vertices to feature space $D$.

\section{Experiments}\label{sec.exp}

\begin{figure*}[!t]
	\centering
	\begin{subfigure}[t]{0.24\textwidth}
		\centering
		\includegraphics[width=1.1\textwidth]{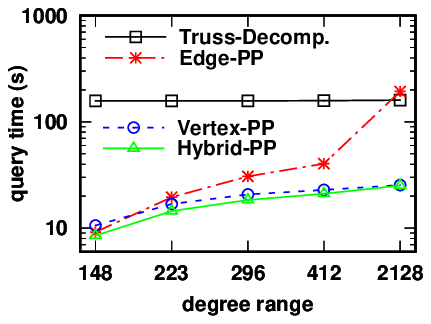}
		\vspace{-0.5cm} 
		\caption{PP-DBLP-2013}
		\label{fig:bydeg13}
	\end{subfigure}
	\begin{subfigure}[t]{0.24\textwidth}
		\centering
		\includegraphics[width=1.1\textwidth]{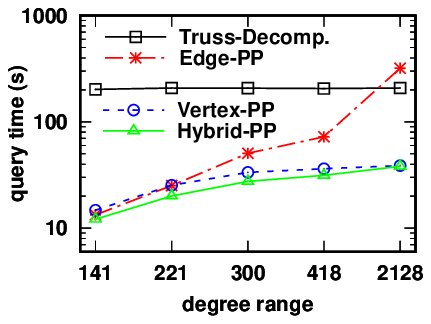}
		\vspace{-0.5cm} 
		\caption{PP-DBLP-2014}
		\label{fig:bydeg14}
	\end{subfigure}
	\begin{subfigure}[t]{0.24\textwidth}
		\centering
		\includegraphics[width=1.1\textwidth]{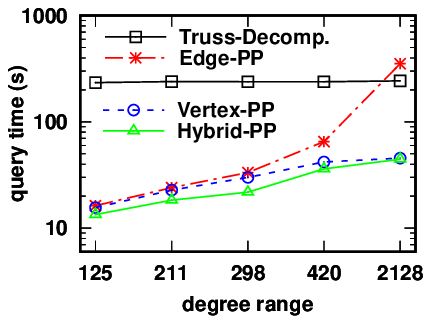}
		\vspace{-0.5cm} 
		\caption{PP-DBLP-2015}
		\label{fig:bydeg15}
	\end{subfigure}
	\begin{subfigure}[t]{0.24\textwidth}
		\centering
		\includegraphics[width=1.1\textwidth]{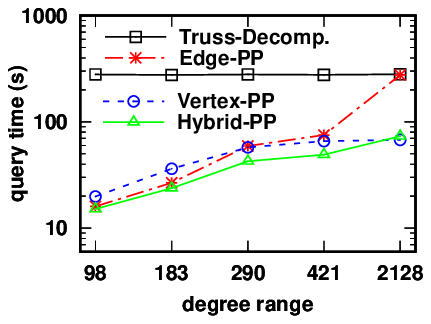}
		\vspace{-0.5cm} 
		\caption{PP-DBLP-2016}
		\label{fig:bydeg16}
	\end{subfigure}

	\vspace{-0.3cm}
	\caption{Average query time of different methods on PP-DBLP varied by query node degree.}\label{fig:dblp-deg}
	\vspace{-0.3cm}
\end{figure*}


\begin{figure*}[!t]
	\centering
	\begin{subfigure}[t]{0.24\textwidth}
		\centering
		\includegraphics[width=1.1\textwidth]{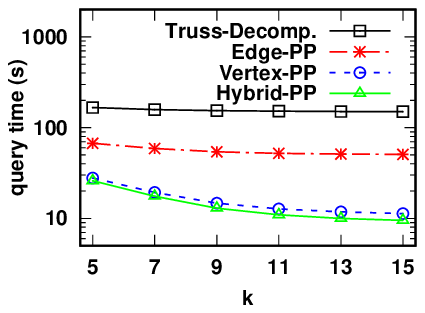}
		\vspace{-0.5cm} 
		\caption{PP-DBLP-2013}
		\label{fig:byk13}
	\end{subfigure}
	\begin{subfigure}[t]{0.24\textwidth}
		\centering
		\includegraphics[width=1.1\textwidth]{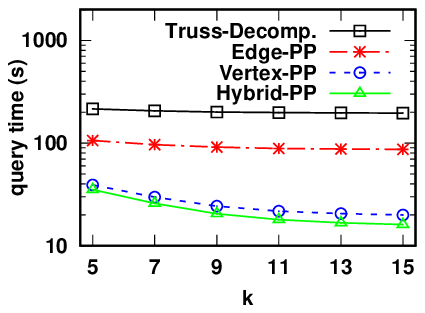}
		\vspace{-0.5cm} 
		\caption{PP-DBLP-2014}
		\label{fig:byk14}
	\end{subfigure}
	\begin{subfigure}[t]{0.24\textwidth}
		\centering
		\includegraphics[width=1.1\textwidth]{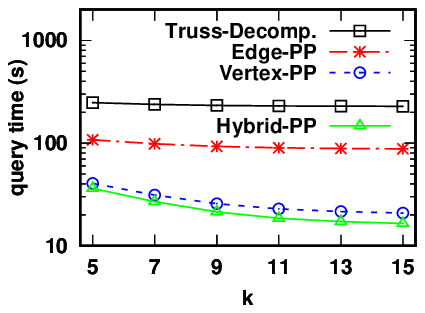}
		\vspace{-0.5cm} 
		\caption{PP-DBLP-2015}
		\label{fig:byk15}
	\end{subfigure}
	\begin{subfigure}[t]{0.24\textwidth}
		\centering
		\includegraphics[width=1.1\textwidth]{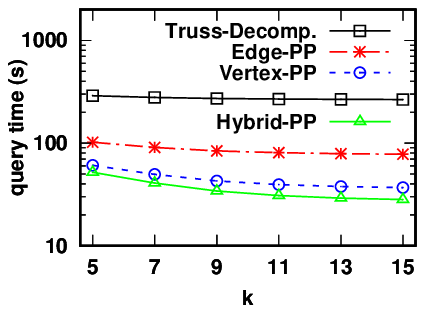}
		\vspace{-0.5cm} 
		\caption{PP-DBLP-2016}
		\label{fig:byk16}
	\end{subfigure}
	\vspace{-0.3cm}
	\caption{Average query time of different methods on PP-DBLP varied by parameter $k$.}\label{fig:dblp-k}
	\vspace{-0.4cm}
\end{figure*}


\stitle{Datasets:} 
We used four public-private graphs from real-world DBLP records called PP-DBLP \cite{huang2018pp}.\footnote{\url{https://github.com/samjjx/pp-data}} Published articles make the public network, and ongoing collaborations form the private networks which are only known by partial authors. Network statistics of PP-DBLP are given in Table~\ref{tab:ppdblp-dataset}.
We also use\soroush{d} 10 real-world graphs available from SNAP \cite{snapnets} shown in Table~\ref{tab:snap}. 

\begin{table}[t]
\scriptsize
\centering
\begin{tabular}{lrrrr}  
\toprule
Name & $|V|$ & $|E|$ & $|V_{private}|$ & $|E_{private}|$\\

\midrule
PP-DBLP-2013	& 1,791,688	 & 5,187,025 & 804,121 & 3,166,863  \\ 
PP-DBLP-2014	& 1,791,688	 & 5,893,083 & 669,138 & 2,491,847  \\
PP-DBLP-2015	& 1,791,688	 & 6,605,428 & 502,654 & 1,719,794  \\
PP-DBLP-2016	& 1,791,688	 & 7,378,090 & 257,129 & 719,204  \\
\bottomrule
\end{tabular}
\vspace{-0.3cm}
\caption{Network Statistics of Real-world Public-Private Graphs}
\vspace{-0.3cm}
\label{tab:ppdblp-dataset}
\end{table}

\begin{table}[t]
\scriptsize
\centering

\begin{tabular}{lrrrrrr}  
\toprule

\multirow{2}{*}{Name} & \multirow{2}{*}{$|V|$ }& \multirow{2}{*}{$|E|$} & \multirow{2}{*}{$|V_{t}|$} & \multicolumn{2}{c}{Avg. time per node (s)}    & \multirow{2}{*}{Speedup}\\ \cmidrule{5-6} 
& & & & Node-ins. & Edge-ins.\\
\midrule
DBLP & 334\textbf{K} & 925\textbf{K} & 940 & 0.05 & 0.15 & 3.20\\
AstroPh & 18\textbf{K} & 396\textbf{K} & 750 & 0.08 & 0.67 & 7.99\\
EmailEuAll & 265\textbf{K} & 420\textbf{K} & 292 & 0.38 & 4.20 & 10.96\\
Wikivote & 7\textbf{K}  & 103\textbf{K} & 534 & 0.73 & 28.53 & 39.10\\
EmailEnron & 36\textbf{K}  & 367\textbf{K} &  485 & 0.83 & 24.10 & 29.03\\
Gowalla & 196\textbf{K} &1.9\textbf{M} & 156 & 2.99 & 102.39 & 34.27\\
WikiTalk & 2.4\textbf{M} & 5\textbf{M} & 191 & 40.54 & 2017.26 & 49.76\\
Flickr & 80\textbf{K} & 11.8\textbf{M} & 697 & 96.30 & 2195.57 & 22.80\\
Digg & 771\textbf{K} & 7.3\textbf{M} & 471 & 142.82 & 5741.43 & 40.20\\
LiveJournal & 4\textbf{M}  & 34.7\textbf{M} & 377 & 142.63 & 5515.66 & 38.67\\
\bottomrule
\end{tabular}
\vspace{-0.3cm}
\caption{Comparison of Node-insertion and Edge-insertion methods in terms of efficiency, by inserting $V_{t}$ randomly selected nodes on real-world graphs. Here $\textbf{K}=10^3$ and $\textbf{M}=10^6$. 
}
\vspace{-0.3cm}
\label{tab:snap}
\end{table}

\stitle{Compared Methods and Evaluated Metrics:} To evaluate the efficiency of improved strategies proposed in this paper, we test\soroush{ed} and compare\soroush{d} four algorithms as follows.

\squishlisttight
\item \baseline:\eat{is} \soroush{an} online search approach using truss decomposition for computing $k$-truss index from scratch \cite{WangC12}.
\item \epp:\eat{is} an approach using the edge-insertion algorithm for updating truss index \cite{huang2014querying}.
\item \vpp:\eat{is} our approach using the node-deletion and node-insertion for updating truss-index in Algorithm \ref{algo:nodeinsertion}.
\item \hybridpp:\eat{is} our hybrid approach using both \epp and \vpp for updating truss-index in Algorithm \ref{algo:hybridpp}.
\end{list} 

After updating the index, all of the four methods mentioned above use the same query method. We compare them by reporting the running time in seconds. The less the running time is, the better the efficiency performance is.
We set the parameter $k=7$ by default. We also evaluate the methods by varying parameters $k$ in $\{5, 7, 9, 11, 13, 15\}$.

\subsection{Efficiency Evaluation on SNAP Networks}
\label{sec.nodeeval}

To evaluate the efficiency of the Node-insertion algorithm, we conducted experiments on 10 SNAP graph datasets in Table \ref{tab:snap}. Due to no available private information in these networks, we randomly generated private edges as follows. We divided nodes into 40 bins by their degree and took 50 randomly selected nodes from each bin. Bin set was defined as $\{B_1, B_2, \ldots B_{40}\}$ where $B_i = \{v : {i - 1 \over 40} < {d(v) \over \Delta} \le {i \over 40} \}$ with $\Delta = \max_{v \in V}\{d(v)\}$. Some bins had less than 50 nodes, and in this case, all nodes of that bin were selected. Let $G = (V, E)$ denote the initial graph, and $V_{t}$ the set of sampled vertices. All edges with at least one end in $V_{t}$ was considered to be private. We ran \baseline on the induced subgraph of vertex set $V' = V \setminus V_{t}$, then added each node $v \in V_{t}$ using Algorithm \ref{algo:nodeinsertion}, and compared the running time with adding edges one by one using edge-insertion algorithm. Experiment results and running time of both algorithms are available in Table \ref{tab:snap}. Obviously, our node-insertion algorithm significantly outperformed the edge-insertion method which achieved 38.67 times speedup on LiveJournal.

\subsection{Classification Evaluation}
In order to choose the proper classification method to incorporate into \hybridpp, we test\soroush{ed} and compare\soroush{d five} classifiers\eat{,} in terms of classification accuracy and training time\soroush{,} as shown in Table \ref{tab:ppdblp-classifier}. Due to the low training time, high classification accuracy, and fast query time, Random Forest \eat{is}\soroush{was} finally used as the classifier of \hybridpp.

\begin{table}[h]
  \scriptsize
  \vspace{-0.2cm}
  \centering
  \begin{tabular}{lrrrrr}  
    \toprule
    
    \multirow{2}{*}{Classifier} & \multicolumn{4}{c}{Accuracy on PP-DBLP} & \multirow{2}{*}{Training}\\ \cmidrule{2-5} 
    & 2013 & 2014 & 2015 & 2016 & time (s)\\
    
    \midrule
    Random Forest & \textbf{84.4\%} & \textbf{87.0\%} & \textbf{86.5\%} & \textbf{87.8\%} & $<$1 \\
    Decision Tree & 84.3\% & 84.6\% & 84.7\% & 87.4\% & \soroush{$<$1} \\
    SVM       & 84.1\% & 84.5\% & 85.5\% & 85.2\% & 1140 \\
    k-NN      & 76.0\% & 79.6\% & 78.2\% & 72.7\% & $<$1 \\
    Degree Baseline & 69.5\% & 69.2\% & 71.8\% & 67.2\% & $<$1 \\
    \bottomrule
  \end{tabular}
  \vspace{-0.3cm}
  \caption{Accuracy of different classifiers on PP-DBLP 2013-2016.}
  \vspace{-0.3cm}
  \label{tab:ppdblp-classifier}
\end{table}

\subsection{Efficiency Evaluation on PP-DBLP Networks}
In this section, we conducted experiments on real-world public-private datasets of PP-DBLP. We compare the efficiency of four different methods \baseline, \epp, \vpp, and \hybridpp. \hybridpp adopted a Random Forest with 51 estimators and a maximum depth of 11 to construct a classifier.
\soroush{The \hybridpp used to answer queries on each PP-DBLP dataset was trained based on runtimes of \vpp and \epp on sampled nodes from the other three datasets.}
\soroush{We first divided all nodes into $100 \times 100$ bins by their public and private degrees, and then randomly took four nodes from each bin.}
Bin set was $\{B_{i, j} : 1 \le i, j \le 100  \}$ each bin defined as $B_{i, j} = \{ v : {i - 1 \over 100} < {d(v) \over \Delta} \le {i \over 100},  {j - 1 \over 100} < {d_p(v) \over \Delta_p} \le {j \over 100} \}$ where $\Delta$ is the maximum public degree, $\Delta_p$ is the maximum private degree, and $d_p(v)$ is the private degree of node $v$. In total $1836 \pm 58$ nodes were selected on each PP-DBLP datasets. 

\stitle{Vary Node Degree.}
We fixed the parameter $k = 7$ and ran the four proposed algorithms on sampled nodes. For better visualization and comparison, sampled nodes were divided into five equally sized groups by node degrees in their pp-graphs,
each group taking $20\%$ of sampled nodes, and the average query time of algorithms on each group is reported in Figure \ref{fig:dblp-deg}. As we can see, the performance of \epp is better than \vpp in lower degree nodes, and the \hybridpp is the fastest as it takes the quicker algorithm between \epp and \vpp in most cases. On higher degree nodes, \epp takes much longer than \vpp, becomes less useful for \hybridpp; hence the gap between \hybridpp and \vpp decreases as \vpp becomes the optimal choice for the classifier on higher degree nodes. 

\stitle{Vary Parameter $k$.} The average query time varied by $k$ is reported in Figure \ref{fig:dblp-k}. \hybridpp and \vpp perform much better than \epp and \baseline for all $k$ values and datasets. \hybridpp is the fastest due to using \epp in cases where \vpp works worse than \epp.


\section{Conclusions}\label{sec.con}
This paper studies the problem of finding k-truss on public-private graphs. We develop a novel hybrid algorithm of k-truss updating with node/edge insertions/deletions, which can incrementally compute k-truss on public-private networks. 
This work opens up several interesting problems. First, developing further efficient and clever algorithms for finding $k$-truss in pp-graphs is important, instead of deleting and re-inserting nodes as \vpp. Second, finding other kinds of dense subgraphs on public-private graph is also wide open.


\bibliographystyle{named}
\bibliography{strucdiv,pp-dblp}

\end{document}